\documentclass[10pt]{iopart}
\usepackage{iopams}
\usepackage{graphicx}
\usepackage{color}

\usepackage{amssymb}
\usepackage{amsthm}
\usepackage{mathrsfs}
\usepackage{bm}

\begin{document}

\newcommand{\REV}[1]{{\color{red}#1}}
\newcommand{\RM}{\mathbb{R}}
\newcommand{\ZM}{\mathbb{Z}}
\newcommand{\QM}{\mathbb{Q}}
\newcommand{\NM}{\mathbb{N}}
\newcommand{\CM}{\mathbb{C}}
\renewcommand{\PM}{\mathbb{P}}
\newcommand{\SM}{\mathbb{S}}
\newcommand{\R}{\mathcal{R}}
\newcommand{\I}{\mathcal{I}}
\def\e{\mathrm{e}}
\def\i{\mathrm{i}}
\def\d{\mathrm{d}}
\newcommand{\MS}{MS}

\newtheorem{theorem}{Theorem}
\newtheorem{lemma}{Lemma}
\newtheorem{remark}{Remark}
\newtheorem{definition}{Definition}

\title{Tomography: mathematical aspects and applications}

\author{Paolo Facchi$^{1,2}$, Marilena Ligab\`o$^3$, Sergio Solimini$^3$}

\address{$^1$Dipartimento di Fisica and MECENAS, Universit\`a di Bari, I-70126 Bari, Italy}
\address{$^2$INFN, Sezione di Bari, I-70126 Bari, Italy} 
\address{$^3$Dipartimento di Meccanica, Matematica e Management, Politecnico di Bari,
        I-70125  Bari, Italy}

\begin{abstract}
In this article we present a  review of the Radon transform and the instability of the tomographic reconstruction process. We show some new mathematical results in tomography obtained by a variational formulation of the reconstruction problem based on the minimization of a Mumford-Shah type functional. Finally, we exhibit a physical interpretation of this new technique and discuss some possible generalizations.\end{abstract}

\pacs{03.65.Wj, 42.30.Wb, 02.30.Uu}

\vspace{2pc}
\noindent{\it Keywords}: Radon transform; integral geometry; image segmentation; calculus of variations.

\maketitle
\ioptwocol

\section{Introduction}
The primary goal of tomography is to determine the internal structure of an object without cutting it, namely  using data obtained by methods that leave the object under investigation undamaged.
These data can be obtained by exploiting the interaction between the object and various kinds of probes including X-rays, electrons, 
and many others. After its interaction with the object under investigation, the probe is detected to produce what we call a projected distribution or tomogram, see Fig. \ref{fig:profiles}.

\begin{figure}[t]
\includegraphics[width=0.45\textwidth]{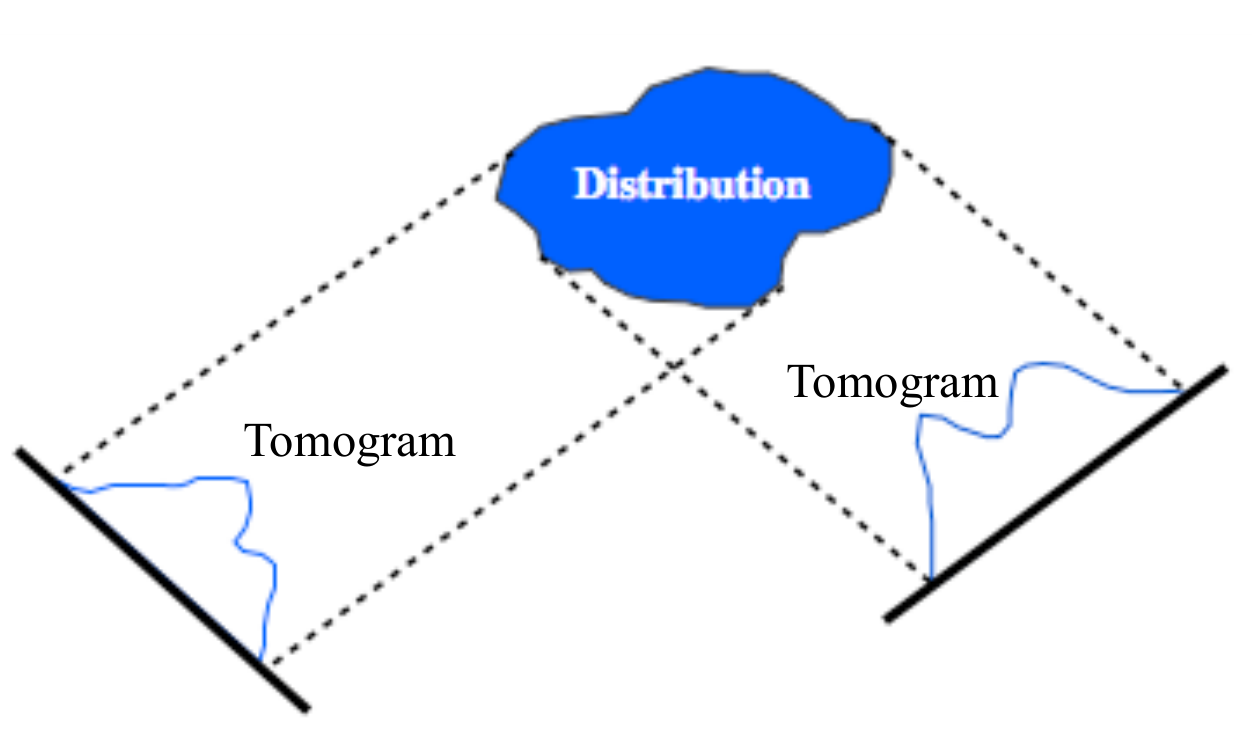}
\caption{Distribution and two tomograms. }
\label{fig:profiles}
\end{figure}
Tomography is a rapidly evolving field for its broad impact on issues of fundamental nature and for its important applications such as the development of diagnostic tools relevant to disparate fields, such as engineering, biomedical and archaeometry. Moreover, tomography can be a powerful tool for many reconstruction problems coming from many  areas of research, such as imaging, quantum information and computation, cryptography, lithography, metrology and many others, see Fig.~\ref{fig:Tomography}.

\begin{figure}[t]
\includegraphics[width=0.45\textwidth]{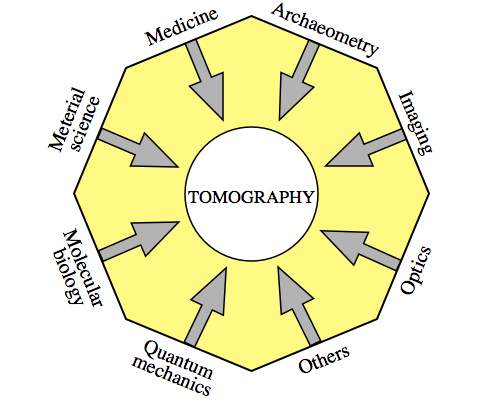}
\caption{Reconstruction problems from diverse fields may be united within the framework of tomography.}
\label{fig:Tomography}
\end{figure}

From the mathematical point of view the reconstruction problem can be formulated as follows: one wants to recover an unknown function through the knowledge of an appropriate family of integral transforms. 
It was proved by J. Radon~\cite{Rad1917} that a smooth function $f(x,y)$ on $\RM^2$ can be determined explicitly by means of its integrals over the lines in $\RM^2$. Let $\R{f}(X,\theta)$ denote the integral of $f$ along the line $x \cos\theta + y \sin\theta  = X$ (tomogram). Then

\begin{equation}
\label{eq:Radoninversion}
f(x,y)=(-\Delta)^{\frac{1}{2}}\int_{0}^{2\pi}   \R{f}(x \cos\theta + y \sin\theta, \theta)\; 
\frac{\d\theta}{4\pi },
\end{equation}
where $\Delta=\frac{\partial^2}{\partial x^2} + \frac{\partial^2}{\partial y^2}$ is the Laplacian on $\RM^2$, and its square root is defined by Fourier transform (see Theorem~\ref{thm:inversioneformula}).
We now observe that the formula above has built in a remarkable duality: first one integrates over the set of points in a line, then one
integrates over the set of lines passing through a given point. This formula can be extended to the $N$-dimensional case by computing the integrals of the function $f$ on all possible hyperplanes.
This suggests to consider the transform $f \mapsto  \R{f}$ defined as follows. If $f$ is a function on $\RM^N$ then $\R{f}$ is the function defined on the space of all possible $(N-1)$-dimensional planes in $\RM^N$ such that, given a hyperplane $\lambda$, the value of $\R{f}(\lambda)$ is given by the integral of $f$ along  $\lambda$. The function $\R{f}$ is called \emph{Radon transform} of $f$.

There exist several important generalizations of the
Radon transform by 
John \cite{John}, Gel'fand \cite{Gelf}, Helgason \cite{Helgason} and Strichartz \cite{Strichartz}.
 More recent analysis has been boosted by Margarita and Volodya Man'ko and has focused on symplectic transforms \cite{Mancini95}, on the deep relationship with classical systems and classical dynamics \cite{Olga97,tomogram,m2}, on the formalism of
star product quantization \cite{MarmoPhysScr,MarmoJPA,Marmoopen}, and on the study of
marginals along curves that are not straight lines
\cite{ManMenPhysD,tomocurved}.

In quantum mechanics the Radon transform of the Wigner function
\cite{Wig32,Moyal,Hillary84}
was considered in the tomographic approach to the study
of quantum states \cite{Ber-Ber,Vog-Ris} and experimentally
realized with different particles and in diverse situations. 
For a review on the modern mathematical aspects of classical and quantum tomography see~\cite{tomolectures}. 
Good reviews on recent tomographic applications can be found in~\cite{Deans} and in~\cite{Jardabook}, where particular emphasis is given on maximum likelihood methods, that enable one to extract the maximum reliable
information from the available data can be found.

As explained above, from the mathematical point of view, the internal structure of the object is described by an unknown function $f$ (density), that is connected via an operator to some measured quantity $g$ (tomograms). The tomographic reconstruction problem can be stated as follows: for given data $g$, the task is to find $f$ from the operator equation $\R{f}=g$. 
There are many problems related to the implementation of effective tomographic techniques due to the instability of the reconstruction process. There are two principal reasons of this instability. The first one is the ill-posedness of the reconstruction problem: in order to obtain a satisfactory estimate of the unknown function it is necessary an extremely precise knowledge of its tomograms, which is in general physically unattainable~\cite{tomothick}. The second reason is the discrete and possibly imperfect nature of data that allows to obtain only an approximation of the unknown function.

The first question is whether a partial information still determines the function uniquely. A negative answer is given by a theorem of Smith, Solomon and Wagner~\cite{SmithSolmonWagner}, that states: ``A function $f$ with compact support in the plane is uniquely determined by any infinite set, but by no finite set of its tomograms''. Therefore, it is clear that one has to abandon the request of uniqueness in the applications of tomography. Thus, due to the ill-posedness of reconstruction problem and to the loss of uniqueness in the inversion process, a regularization method has to be introduced to stabilize the inversion.

A powerful approach is the introduction of a Mumford-Shah (MS)  functional, first introduced in a different context for image denoising and segmentation~\cite{MS}. The main motivation is that, in many practical applications, one is not only interested in the reconstruction of the density distribution $f$, but also in the extraction of some specific features or patterns of the image. An example is the problem of the determination of the boundaries of inner organs. By minimizing the MS functional, one can find not only (an approximation of) the function but also its sharp contours. Very recently a MS functional for applications to tomography has been introduced in the literature~\cite{RamlauRing1, RamlauRing2, RondiSantosa}. Some preliminary results in this context are already available but there are also many interesting open problems and promising results in this direction, as we will try to explain in the second part of this article.

The article is organized as follows. Section~\ref{Sec:Radon} contains a short introduction to the Radon transform, its dual map and the inversion formula. Section~\ref{Sec:ill} is devoted to a brief  discussion on the ill-posedness of the tomographic reconstruction and to the introduction of  regularization methods. In Section~\ref{Sec:MS} a MS functional is applied to tomography  as  a regularization method. In particular, in Subsection~\ref{Subsec:MS}  the piecewise constant model and  known results are discussed together with a short list of some interesting open problems.
Finally, in Section~\ref{Sec:3Dinterpretation} we present an electrostatic  interpretation of the regularization method based on the MS functional,
which motivates us to introduce an improved regularization method, based on the Blake-Zisserman functional~\cite{BZ}, as a relaxed version of the previous one.

\section{The Radon transform: definition and inversion formula}\label{Sec:Radon}

Consider a body in the plane $\RM^2$, and consider a beam of particles (neutrons, electrons, X-rays, etc.) emitted by a source. Assume that the initial intensity of the beam is $I_{0}$. When the particles pass through the body they are absorbed or scattered and the intensity of the beam traversing a length $\Delta s$ decreases by an amount  proportional to the density of the body $\mu$, namely
\begin{equation}
\Delta I /I=-\Delta s\, \mu(s),
\end{equation}
so that
\begin{equation}
I(s)=I_{0}\exp \left(-\int_{0}^{s}\mu(r)\;\d r \right).
\end{equation}
A detector placed at the exit of the body measures the final intensity $I(s)$ and then from
\begin{equation}
-\ln \frac{I(s)}{I_0} =\int_{0}^{s} \mu(r)\; \d r
\end{equation}
one can record the value of the density $\mu$ integrated on a line.
If another ray with a different direction is considered, with the same procedure one obtains the value of the integral of the density on that line. 

The mathematical model of the above setup is the following:
Given a 
smooth function $f(x)$ on the plane, $x\in\RM^2$,
and a line  $\lambda$, consider its tomogram, given by
\begin{equation}\label{sharp lambda}
\R{f}(\lambda)=\int_{\lambda}  f(x)\; \d m(x),
\end{equation}
where $\d m$ is the Euclidean measure on the line $\lambda$.
In this way, we have defined an operator $\R$ that maps a smooth function $f$ on the plane $\RM^2$ into a function $\R{f}$ on $\PM^{2}$, the manifold of the lines in $\RM^2$.

We ask the following question: If we know the family of tomograms $(\R{f}(\lambda))_{\lambda \in \PM^{2}}$, can we reconstruct the density function $f$? The answer is affirmative and in the following we will see how to obtain this result.

Let us generalize the above definitions to the case of an $N$-dimensional space.
Let $f$ be a function defined on $\RM^N$, integrable on each hyperplane in $\RM^N$ and let $\PM^N$ be the manifold of all hyperplanes in $\RM^N$. The Radon transform of $f$ is defined 
by Eq.~(\ref{sharp lambda}),
where $\d m$ is the Euclidean measure on the hyperplane $\lambda$.
Thus we have an operator $\R$, the \emph{Radon transform}, that maps a function $f$ on $\RM^N$ into a function $\R{f}$ on $\PM^{N}$, namely $f \mapsto \R{f}$. Its dual  transform, also called \emph{back projection operator}, $g \mapsto \I{g}$ associates to a function $g$ on $\PM^{N}$ the function $\I{g}$ on $\RM^N$ given by
\begin{equation}\label{dual map}
\I{g}(x)=\int_{x \in \lambda} g(\lambda) \; \d\mu(\lambda),
\end{equation}
where $\d\mu$ is the unique probability measure on the compact set $\{\lambda \in \PM^{N}|x \in \lambda\}$ which is invariant under the group of rotations around $x$.

\begin{figure}[t]
\includegraphics[width=0.45\textwidth]{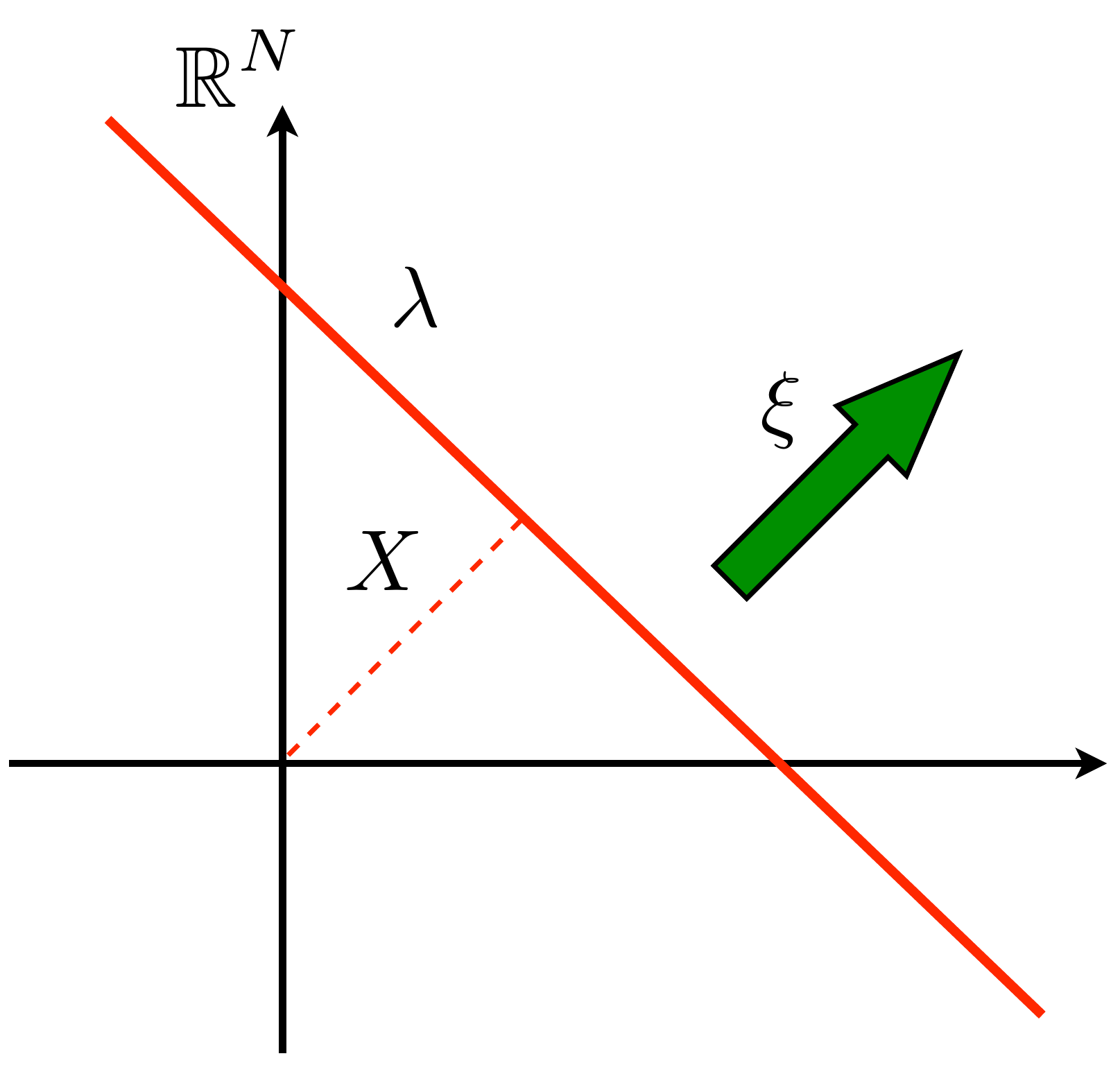}
\caption{Parametrization of the hyperplane $\lambda$ using its signed distance $X$ from the origin and a unit vector $\xi$ perpendicular to $\lambda$.}
\label{fig:radon nd}
\end{figure}

Let us consider the following covering of $\PM^N$
\begin{equation}\label{double covering}
\RM \times \SM^{N-1} \to \PM^N, \qquad (X,\xi)\mapsto \lambda,
\end{equation}
where $\SM^{N-1}$ is the unit sphere in $\RM^N$. Thus, the
equation of the hyperplane $\lambda$ is
\begin{equation}
\lambda=\{x\in\RM^N \, | \,  X-\xi \cdot x=0\},
\end{equation}
with $a\cdot b$ denoting the Euclidean inner product of $a, b\in\RM^N$. See Fig. \ref{fig:radon nd}.
Observe that the pairs $(X,\xi),(-X,-\xi) \in \RM \times \SM^{N-1}$ are mapped into the same hyperplane $\lambda \in \PM^N$. Therefore (\ref{double covering}) is a double covering of $\PM^N$.
Thus $\PM^N$ has a canonical manifold structure with respect to which this covering mapping is differentiable. We  identify continuous (differentiable) functions $g$ on $\PM^N$ with continuous (differentiable) functions $g$ on $\RM \times \SM^{N-1}$ satisfying $g(X,\xi)=g(-X,-\xi)$.

We will momentarily work in the Schwartz space $\mathcal{S}(\RM^N)$ of complex-valued rapidly decreasing functions on $\RM^N$. 
In analogy with $\mathcal{S}(\RM^N)$ we define $\mathcal{S}(\RM \times \SM^{N-1})$ as the space of $C^{\infty}$ functions $g$ on $\RM \times \SM^{N-1}$ which for any integers $m\geq 0$, any  multiindex $\alpha \in \NM^N$, and any differential operator $D$ on $\SM^{N-1}$ satisfy
\begin{equation}
\sup_{X \in \RM , \xi \in \SM^{N-1}}\left| |X|^m \frac{\partial^\alpha(Dg)}{\partial X^\alpha}(X,\xi)\right| < +\infty.
\end{equation}
The space $\mathcal{S}(\PM^N)$ is then defined as the set of $g \in \mathcal{S}(\RM \times \SM^{N-1})$ satisfying $g(-X,-\xi)=g(X,\xi)$.

Now we want to obtain an inversion formula, namely we want to prove that one can recover a function $f$ on $\RM^N$ from the knowledge of its Radon transform.
 In order to get this result we need a preliminary lemma, whose proof can be found in~\cite{tomolectures}, which suggests an interesting physical interpretation. 
 \begin{lemma}\label{sharp flat}
Let $f \in \mathcal{S}(\RM^N)$ and $V(x)=1/|x|$, $x \in \RM^N$, $x\neq0$ . Then
\begin{equation}
\I(\R{f})(x)=a_{N}\, f\ast V, 
\label{eq:electid}
\end{equation}
where $a_N$ depends only on the dimension $N$, and
$*$ denotes the convolution product,
\begin{equation}
(f\ast g)(x) =\int_{\RM^N} f(y)\, g(x-y) \; \d y  .
\end{equation}
\end{lemma}

\begin{figure}[t]
\includegraphics[width=0.45\textwidth]{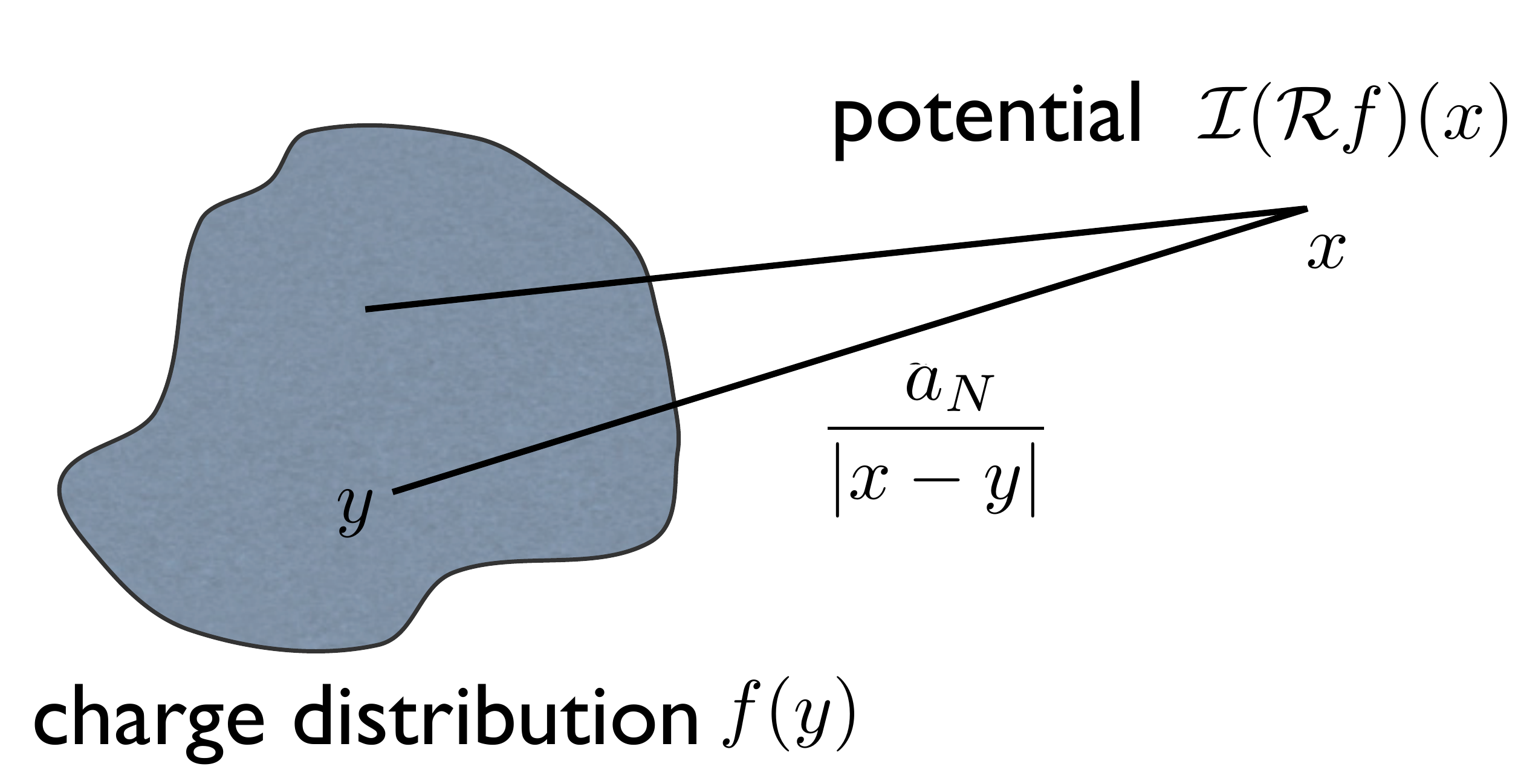}
\caption{$\I(\R{f})(x)$ is the potential at $x$ generated by the charge distribution $f$.}
\label{fig:potential}
\end{figure}

A  physical interpretation of Lemma~\ref{sharp flat} is the following: if $f$ is a charge distribution, then the potential at the point $x$ generated by that charge is exactly $\I(\R{f}(x))$, see Fig.~\ref{fig:potential}. Notice, however, that the potential of a point charge scales always as the inverse distance \emph{independently} of the dimension $N$, and thus it is Coulomb only for $N=3$. The only dependence on $N$ is in the strength of the elementary charge $a_N$. This fact is crucial: indeed, the associated Poisson equation involves an $N$-dependent (fractional) power of the Laplacian, which appears in the inversion formula for the Radon transform.
\begin{theorem}\label{thm:inversioneformula}
Let $f \in \mathcal{S}(\RM^N)$. Then
\begin{equation}\label{inversion formula0}
f(x)= \frac{1}{2(2\pi)^{N-1}}(-\Delta)^{\frac{N-1}{2}} \I(\R{f})(x)
\end{equation}
where $(-\Delta)^\alpha$, with $\alpha>0$, is a pseudodifferential operator whose action is
\begin{equation}
((-\Delta)^\alpha f) (x)= \int_{\RM^N} |k|^{2\alpha}\, \hat{f}(k)\,  \e^{\i k \cdot x}\; \frac{\d k}{(2\pi)^N},
\label{eq:fractionalLap}
\end{equation}
where $\hat{f}$ is the Fourier transform of $f$,
\begin{equation}
\hat{f}(k) = \int_{\RM^N} f(x)\,  \e^{-\i k \cdot x}\; \d x.
\label{eq:Fouriertrans}
\end{equation}
\end{theorem}
The proof of Theorem~\ref{thm:inversioneformula} can be found in~\cite{tomolectures}.
Equation~(\ref{inversion formula0}) says that, modulo the final action of $(-\Delta)^{\frac{N-1}{2}}$, the function $f$ can be recovered from its Radon transform $\R{f}$ by the application of the dual mapping $\I$: first one integrates over the set of points in a hyperplane and then one integrates over the set of hyperplanes passing through a given point. Explicitly we get
\begin{equation}\label{inversion formula}
f(x) =\frac{1}{2(2\pi)^{N-1}}(-\Delta)^{\frac{N-1}{2}}\int_{\SM^{N-1}}   \R{f}(\xi \cdot x, \xi)\; \d\xi,
\end{equation}
which has the following remarkable interpretation. Note that if one fixes a direction $\xi \in \SM^{N-1}$, then the function $\R{f}(\xi \cdot x,\xi)$  is constant on each plane perpendicular to $\xi$, i.e.\ it is a (generalized) plane wave. Therefore, Eq.~(\ref{inversion formula}) gives a representation of $f$ in terms of a continuous superposition of plane waves. A well-known analogous decomposition is given by Fourier transform. When $N=2$, one recovers the inversion formula~(\ref{eq:Radoninversion}) originally found by Radon \cite{Rad1917}.

\section{Instability of the inversion formula with noisy data}\label{Sec:ill}
We have defined the Radon transform of any function $f  \in \mathcal{S}(\RM^N)$ as $\R{f}$. The following theorem~\cite{Helgason}
 contains the characterization of the range of the Radon linear operator $\R$ and the extension of $\R$ to the space of square integrable functions $L^2(\RM^N)$.
\begin{theorem}
\label{th:Radonbijection}
The Radon transform $\R$ is a linear one-to-one mapping of $\mathcal{S}(\RM^N)$ onto $\mathcal{S}_H(\PM^N)$, where the space  $\mathcal{S}_H(\PM^N)$ is defined as follows: $g \in \mathcal{S}_H(\PM^N)$ if and only if $g \in \mathcal{S}(\PM^N)$ and for any integer $k \in \NM$ the integral
\begin{equation}
\int_{\RM} g(X,\xi)\, X^k \; \d X
\end{equation}
is a homogeneous polynomial of degree $k$ in $\xi_1, \dots, \xi_N$. Moreover, the Radon operator $\R$ can be extended to a continuous operator from $L^2(\RM^N)$ and $L^2(\PM^{N})$.
\end{theorem}

In medical imaging, computerized tomography is a widely used technique
for the determination of the  density $f$ of a sample from measurements
of the attenuation of X-ray beams sent through the material along
different angles and offsets. The measured data $g$ are connected to the density
$f$ via the Radon Transform $\R{}$.
To compute the density distribution $f$ the equation $g = \R{f}$
has to be inverted. Unfortunately it is a well known fact that $\R$ is not continuously invertible on $L^2(\PM^{N})$~\cite{Helgason}, and this imply that the problem of inversion is ill-posed.
For this reason, regularization methods have to be introduced to stabilize the inversion in the presence of data noise.

We discuss ill-posed problems only in the framework of linear problems in
Hilbert spaces~\cite{Natterer}. Let $\mathscr{H}, \mathscr{K}$ be Hilbert spaces and let $A$ be a linear bounded operator
from $\mathscr{H}$ into $\mathscr{K}$. The problem
\begin{equation}\label{prob:inverseproblem}
\textrm{given $g\in \mathscr{K}$, find $f \in \mathscr{H}$ such that $Af=g$}
\end{equation}
is called well-posed by Hadamard (1932) if it is uniquely solvable for each $g \in  \mathscr{K}$
and if the solution depends continuously on $g$. Otherwise, (\ref{prob:inverseproblem}) is called ill-posed.
This means that for an ill-posed problem the operator $A^{-1}$ either does not exist,
or is not defined on all of $\mathscr{K}$, or is not continuous. The practical difficulty with an
ill-posed problem is that even if it is solvable, the solution of $Af = g$ need not be
close to the solution of $Af = g^\epsilon$ if $g^\epsilon$ is close to $g$. 

In general $A^{-1}$ is not a continuous operator. To restore continuity we introduce the notion of a regularization of $A^{-1}$. This is a family $(T_\gamma )_{\gamma >0}$ of linear continuous operators $T_\gamma: \mathscr{K} \to  \mathscr{H}$ which are defined on all  $\mathscr{K}$ and for which
\begin{equation}
\lim_{\gamma \to 0} \; \;T_\gamma g= A^{-1} g
\end{equation}
on the domain of $A^{-1}$. Obviously $\|T_\gamma\| \to + \infty$ as $\gamma \to 0$ if $A^{-1}$ is not bounded. With the help of a regularization we can solve (\ref{prob:inverseproblem}) approximately in the following sense. Let $g^\epsilon \in  \mathscr{K}$ be an approximation to $g$ such that $\|g -g^\epsilon\| \leq \epsilon$ . Let $\gamma(\epsilon)$ be such that, as $\epsilon \to 0$,
\begin{equation}
\gamma(\epsilon) \to 0, 
\qquad  \|T_{\gamma(\epsilon)}\| \epsilon \to 0.
\end{equation}
Then, as $\epsilon \to 0$,
\begin{eqnarray}
\|T_{\gamma(\epsilon)} g ^\epsilon -A^{-1} g\|  & \leq & \|T_{\gamma(\epsilon)} (g ^\epsilon-g) \|  
\nonumber\\ & & 
+ \| (T_{\gamma(\epsilon)}-A^{-1}) g\| \nonumber \\
                                                                      & \leq &   \|T_{\gamma(\epsilon)}\| \epsilon +  \| (T_{\gamma(\epsilon)}-A^{-1} )g\|  \to 0. \nonumber \\
\end{eqnarray}
Hence, $T_{\gamma(\epsilon)} g^\epsilon$ is close to $A^{-1} g$ if $g^\epsilon$ is close to $g$. The number $\gamma$ is called a \emph{regularization parameter}. Determining a good
regularization parameter is one of the crucial points in the application of regularization methods.

There are several methods for constructing a regularization as the truncated singular value decomposition, the method of Tikhonov-Phillips  or  some iterative methods~\cite{Natterer}. In the following section we present a regularization method based on the minimization of a Mumford-Shah type functional.

\section{Mumford-Shah functional for the simultaneous segmentation and reconstruction of a function}\label{Sec:MS}

In many practical applications one is not only interested in the reconstruction
of the density distribution $f$ but also in the extraction of some specific
features within the image which represents the density distribution of the
sample. For example, the planning of surgery might require the determination
of the boundaries of inner organs like liver or lung or the separation
of cancerous and healthy tissue.
Segmenting a digital image means finding its \emph{homogeneous regions} and its \emph{edges}, or \emph{boundaries}. Of course, the homogeneous regions are supposed to correspond to meaningful parts of objects in the real world, and the edges to their apparent contours. The Mumford-Shah variational model is one of the principal models of image segmentation. It defines the segmentation problem as a joint smoothing/edge detection problem: given an image $g(x)$, one seeks simultaneously a ``piecewise smoothed image'' $u(x)$ with a set $\Gamma$ of abrupt discontinuities, the ``edges'' of $g$. 
The original Mumford-Shah functional~\cite{MS}, is the following:
\begin{eqnarray}\label{defn:JMS}
\MS(\Gamma,u)&=&\|u-g\|_{L^2(D)}^2+\alpha \int_{D \setminus \Gamma} |\nabla u(x)|^2\; \textrm{d}x \nonumber \\
                       && + \beta \mathcal{H}^{N-1}(D \cap \Gamma),
\end{eqnarray}
where 
\begin{itemize}
\item $D \subset \RM^N$ is an open set (\emph{screen});
\item $\Gamma \subset \RM^N$ is a closed set (\emph{set of edges});
\item $u: D \to \RM$ (\emph{cartoon});
\item $\nabla u$ denotes the distributional gradient of $u$;
\item $g \in L^2(D)$ is the datum (\emph{digital image});
\item $\alpha, \beta >0$ are parameters (\emph{tuning parameters});
\item $\mathcal{H}^{N-1}$ denotes the $(N-1)$-dimensional Hausdorff measure. 
\end{itemize}
The squared $L^2$ distance in (\ref{defn:JMS}) plays the role of a fidelity term: it imposes that the cartoon $u$ approximate the image $g$. The second  term in the functional imposes that the cartoon $u$ be piecewise smooth outside the edge set $\Gamma$. In other word this term  favors sharp contours rather than zones where a thin layer of gray is used to pass smoothly from white to black or viceversa.  Finally the third term in the functional imposes that the contour $\Gamma$ be ``small''  and as smooth as possible. What is expected from the minimization of this functional is a sketchy, cartoon-like version of the given  image together with its contours. See Fig.~\ref{fig:cartooneye}.

The minimization of the $MS$ functional represents a compromise between accuracy and segmentation. The compromise depends on the tuning parameters $\alpha$ and $\beta$
which  have different roles. The parameter $\alpha$ determines how much the cartoon $u$ can vary, if $\alpha$ is small some variations of $u$ are allowed, while as $\alpha$ increases $u$ tends to be a piecewise constant function. The parameter $\beta$ represents a scale parameter of the functional and measure the amount of contours: if $\beta$ is small, a lot of edges are allowed and we get a fine segmentation. As $\beta$ increases, the segmentation gets coarser. 
For more details on the model see the original paper~\cite{MS}, and the book~\cite{JMMS}.

\begin{figure}[t]
\includegraphics[width=0.45\textwidth]{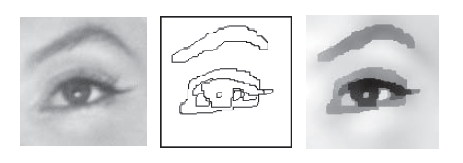}
\caption{Left:  image of an eye ($g$). Center: contours of the image in the Mumford-Shah model (edges $\Gamma$). Right: piecewise smooth function approximating the image (cartoon $u$) \cite{Yuille}.}
\label{fig:cartooneye}
\end{figure}
The minimization of the $\MS$ functional  in (\ref{defn:JMS}) is performed among the admissible pairs $(\Gamma,u)$ such that $\Gamma$ is closed and $u \in C^1(D \setminus \Gamma)$. It is worth noticing that in this model there are two unknowns: a scalar function $u$ and the set $\Gamma$ of its discontinuities. For this reason this category of problems is often called ``Free Discontinuities Problem''. Existence of minimizers of the $\MS$ functional  in~(\ref{defn:JMS}) was proven by De Giorgi, Carriero, Leaci in~\cite{DGCL} in the framework of bounded variation functions without Cantor part (space SBV) introduced by Ambrosio and De Giorgi in~\cite{DGA}.  Further regularity properties for optimal segmentation in the Mumford-Shah model were shown in~\cite{DMMS,AmbFusPal,LMS, BD}.

Here we present a variation of the MS functional, adapted to the inversion problem of the Radon transform. More precisely, we consider a regularization method that quantifies the edge sets together with images, i.e. a procedure  that  gives simultaneously a reconstruction and a segmentation of $f$ (assumed to be  supported in $D\subset\RM^N$) directly from the measured tomograms $g$, based on the minimization of the Mumford-Shah type functional 
\begin{eqnarray}\label{def:MSfunctional}
J_{\MS}(\Gamma,f)&=&\|\R{f}-g\|_{L^2(\PM^N)}^2 \nonumber \\
                          &  &+\alpha \int_{D \setminus \Gamma} |\nabla f(x)|^2\; \textrm{d}x
 + \beta \mathcal{H}^{N-1}(\Gamma).
 \end{eqnarray}
The only difference between the functionals $\MS$ and $J_{\MS}$ is the first term, i.e. the fidelity term, that ensures that the reconstruction for  $f$ is close enough to a solution of the equation $\R{f} = g$, whereas the other terms play exactly the same role explained for the functional $\MS$. As  explained above, in addition to the reconstruction of the density $f$, we are  interested in the reconstruction of its singularity set $\Gamma$,
i.e. the set of points where the solution $f$ is discontinuous. The main difference with respect to the standard Mumford-Shah functional~(\ref{defn:JMS}) is that we have to translate the information about the set of sharp discontinuities of $g$ (and hence on the space of the Radon transform) into information about the strong discontinuities of $f$.

\subsection{The piecewise constant model}\label{Subsec:MS}

Here we will review the results obtained by Ramlau and Ring~\cite{RamlauRing1} concerning the minimization of  (\ref{def:MSfunctional})  restricted to piecewise constant functions $f$, and then consider some interesting open problems.
For medical applications, it is often a good approximation to
restrict the reconstruction to densities $f$ that are constant with respect to a
partition of the body, as the tissues of inner organs, bones, or muscles have
approximately constant density.

We introduce the space $PC_m(D)$ as the space of piecewise constant functions that attain at most $m$ different function values, where $D$ is an open and bounded subset of $\RM^N$. In other words, each $f \in PC_m(D)$ is a linear combination of $m$ characteristic
functions $\chi_{\Omega_k}$ of sets $(\Omega_k)_{k=1,\dots,m}$ which satisfy
$$
\sum_{k=1}^m\chi_{\Omega_k}=\chi_{D} \qquad \mathrm{a.e.}
$$
We assume that the $\Omega_k$'s are open relatively to $D$ and we set $\Gamma_k=\partial \Omega_k$ for the boundary of $\Omega_k$ with respect to the topology relative to the open domain $D$.
In this situation the edge set will be given by the union of the boundaries of $\Omega_k$'s.  For technical reasons it is necessary to assume a \emph{nondegeneracy} condition on the admissible partitions of $D$:  
\begin{equation}\label{nondegeneracy-cond}
\bm{\Omega}=(\Omega_k)_{k=1,\dots,m}\quad \textrm{is admissible if}\quad  \mathcal{L}^N(\Omega_k)\geq \delta,
\end{equation}
for some $\delta >0$, for all $k=1,\dots,m$, where $\mathcal{L}^{N}$ denotes the Lebesgue measure on $\RM^N$.

It turns out to be convenient to split the information encoded in a typical function,
\begin{equation}
f=\sum_{k=1}^m f_k\chi_{\Omega_k} \in PC_m(D),
\end{equation}
into a ``geometrical'' part described by the $m$-tuple of pairwise disjoint
sets $\bm{\Omega}=(\Omega_k)_{k=1,\dots,m}$ which cover $D$ up to a set of measure zero and a ``functional'' part given by the $m$-tuple of values $\bm{f}=(f_k)_{k=1,\dots,m}$. We also use the notation  $\bm{\Gamma}=(\Gamma_k)_{k=1,\dots,m}$, for the boundaries $\Gamma_k= \partial \Omega_k$ of $\Omega_k$.

As usual when dealing with inverse problems, we have to assume that the data $g$ are not exactly known, but that we are only given noisy measured tomograms $g^\epsilon$ of a (hypothetical) exact data set $g$ with $\|g^\epsilon-g\|_{L^2(\PM^N)} \leq \epsilon$.

If we restrict the functional (\ref{def:MSfunctional}) to functions in $PC_m(D)$ we obtain that the second term (involving the derivatives of $f$) disappears, therefore it remains to minimize the functional
\begin{equation}\label{functional:Jbeta}
J_\beta(\bm{\Omega},\bm{f})=\|\R{f}-g^\epsilon\|^2_{L^2(\PM^N)}+\beta \sum_{k=1}^m \mathcal{H}^{N-1}(\Gamma_k),
\end{equation}
over $PC_m(D)$, with respect to the functional variable $\bm{f}$ (a vector of $m$ components) and the geometric variable $\bm{\Omega}$ (a partition of the domain $D$ with at most $m$ distinct regions satisfying the non degeneracy condition (\ref{nondegeneracy-cond})).
So the problem is to find $\tilde{f} \in PC_m(D)$ such that
\begin{equation}\label{eqn:minimizerJbeta}
\tilde{f}=\sum_{k=1}^m \tilde{f}_k\chi_{\tilde{\Omega}_k} \in PC_m(D),
\end{equation}
where
\begin{equation}\label{argmin:Jbeta}
(\tilde{\bm{\Omega}},\tilde{\bm{f}})  = \textrm{arg}  \min_{(\bm{\Omega} , \bm{f})}\; J_\beta(\bm{\Gamma},\bm{f}).
\end{equation}
It is clear that $\tilde{f}$ will depend on  the regularization parameter $\beta$ and on the error level $\epsilon$.

Now we can state the results concerning the functional $J_\beta$ in (\ref{functional:Jbeta}). There are several  technical details necessary for the precise statement and proof of the theorems, for which we refer to the original paper~\cite{RamlauRing1}.
Here we will give a simplified version of the theorems with the purpose of explain the main goal, without too many technical details.  The first result is about  the existence of minimizers of the functional $J_\beta$ in  (\ref{functional:Jbeta}).
\begin{theorem}
For all $g^\epsilon \in L^2(\PM^N)$ there exists a  minimizer  $(\tilde{\bm{\Omega}}^\epsilon_\beta,\tilde{\bm{f}}_\beta^\epsilon)$ of the functional $J_\beta$ in (\ref{functional:Jbeta}), with $\beta>0$.
\end{theorem}
The second result regards the stable dependence of the minimizers of the  functional $J_\beta$ in (\ref{functional:Jbeta}) on the error level $\epsilon$.
\begin{theorem}
Let $(g^{\epsilon_n})_{n \in \NM}$ be a sequence of functions in $L^2(\PM^N)$ and let $g^\epsilon \in L^2(\PM^N)$. For all $n \in \NM$, let $(\tilde{\bm{\Omega}}^{\epsilon_n}_\beta,\tilde{\bm{f}}_\beta^{\epsilon_n})$ denote the minimizers of the functional $J_\beta$  with initial data $g^{\epsilon_n}$. If $g^{\epsilon_n} \to g^\epsilon$ in $L^2(\PM^N)$, as $n \to +\infty$, then there exists a subsequence of $(\tilde{\bm{\Omega}}^{\epsilon_n}_\beta,\tilde{\bm{f}}_\beta^{\epsilon_n})$ such that 
$$
(\tilde{\bm{\Omega}}^{\epsilon_{n_j}}_\beta,\tilde{\bm{f}}_\beta^{\epsilon_{n_j}}) \to (\tilde{\bm{\Omega}}^\epsilon_\beta,\tilde{\bm{f}}_\beta^\epsilon) 
$$
as $j \to + \infty$, and $(\tilde{\bm{\Omega}}^\epsilon_\beta,\tilde{\bm{f}}_\beta^\epsilon) $ is a minimizer of $J_\beta$ with initial data $g^\epsilon$. Moreover, the limit of each convergent subsequence of $(\tilde{\bm{\Omega}}^{\epsilon_n}_\beta,\tilde{\bm{f}}_\beta^{\epsilon_n})$ is a minimizer of $J_\beta$ with initial data $g^\epsilon$.
\end{theorem}

Finally the last Theorem is a regularization result. 
\begin{theorem}
Let $f^\ast \in PC_m(D)$ be given,
$$
f^\ast=\sum_{k=1}^m f_k^\ast \chi_{\Omega_k^\ast},
$$
and let $g^\ast=\R{f}^\ast$. Assume we have noisy data $g^\epsilon \in L^2(\PM^N)$ with $\|g^\epsilon-g^\ast\|_{L^2(\PM^N)}\leq \epsilon$. Let us choose the parameter $\beta=\beta(\epsilon)$ satisfying the conditions $\beta(\epsilon) \to 0$ and $\epsilon^2/ \beta(\epsilon) \to 0$ as $\epsilon \to 0$. For any sequence $\epsilon_n \to 0$, let 
$(\tilde{\bm{\Omega}}^{n},\tilde{\bm{f}}^{n})$
denote the minimizers of the  functional $J_{\beta(\epsilon_n)}$  with initial data $g^{\epsilon_n}$ and  regularization parameter $\beta=\beta(\epsilon_n)$. Then  there exists a convergent subsequence of 
$(\tilde{\bm{\Omega}}^{n},\tilde{\bm{f}}^{n})$. Moreover, for every convergent subsequence with limit $(\tilde{\bm{\Omega}},\tilde{\bm{f}})$ the function 
$$
\tilde{f}=\sum_{k=1}^m \tilde{f}_k \chi_{\tilde{\Omega}_k} \in PC_m(D)
$$
is a solution of the equation $\R{f}=g^\ast$ with a minimal perimeter. Moreover if $f^\ast$ is the unique solution of this equation then the whole sequence converges
$$
\tilde{f}^{n}\to  f^\ast, \quad \textrm{in $L^2(D)$},
$$
when $n\to+\infty$.
\end{theorem}

Finally, let us list some open problems in this context:
\begin{itemize}
\item Is the nondegeneracy condition (\ref{nondegeneracy-cond}) necessary?
\item Can one find an a priori optimal value for the number $m$ of different values?
\item Is it possible to give an a priori estimate on the $L^{\infty}$-norm of the solution (maximum principle)?
\item And finally, it would be very important for applications to prove the existence of minimizers of the functional $J_{\MS}$ not restricted to piecewise constant functions $f$.
\end{itemize}
We observe that all these problems are quite natural, and have been completely solved in the case of the standard Mumford-Shah functional $\MS$ in (\ref{defn:JMS}), see e.g.~\cite{JMMS,Fusco}.

\section{Electrostatic interpretation of $J_{\MS}$}\label{Sec:3Dinterpretation}

In this section we restrict our attention to the $3$-dimensional case. We propose an electrostatic interpretation of the regularization method based on the functional $J_{\MS}$ discussed in the previous section. The intent is to give a physical explanation of the fidelity term $\|\R{f}-g\|^2_{L^2(\PM^3)}$ in the functional (\ref{def:MSfunctional}), that provide the intuition for an improved regularization method.
For $N=3$, the inversion formula~(\ref{inversion formula0}) and the electrostatic identity~(\ref{eq:electid}) particularize, respectively, as follows: for all $f \in \mathcal{S}(\RM^3)$ one gets
\begin{equation}\label{inversionformula_again}
f=-\frac{1}{2(2\pi)^2}\Delta\I(\R{f})
\end{equation}
and 
\begin{equation}
\I(\R{f})=a_3 f \ast V, 
\end{equation}
where $V(x)=1/|x|$ and $a_3$ is a constant. 
We present two preliminary Lemmas.
\begin{lemma}\label{lemma:normRf}
For all real valued $f \in \mathcal{S}(\RM^3)$ one has
\begin{equation}
\|\R{f}\|^2=\frac{1}{2(2\pi)^2}\|\nabla\I(\R{f})\|^2.
\end{equation}
\end{lemma}
\begin{proof}
We know  that $f=-\frac{1}{2(2\pi)^2} \Delta \I(\R{f})$, therefore
\begin{eqnarray}
&& \int_{\RM \times \SM^2}| \R f(X , \xi)|^2 \; \textrm{d}X\; \textrm{d}\xi  
=\int_{\RM^3} f(x)\, \I(\R{f}) (x)\; \textrm{d}x \nonumber \\
&&=   \int_{\RM^3} \frac{1}{2(2\pi)^2}[-\Delta \I(\R{f})] (x)\, \I(\R{f}) (x) \; \textrm{d}x \nonumber \\
&&=\frac{1}{2(2\pi)^2}  \int_{\RM^3}| \nabla \I(\R{f}) (x)|^2\; \textrm{d}x . 
\nonumber
\end{eqnarray}
\end{proof}

\begin{lemma}\label{lemma:electricfield}
For all real valued $f \in \mathcal{S}(\RM^3)$
define 
\begin{equation}
E=- \frac{1}{2(2\pi)^2} \nabla \I(\R{f})
\end{equation} 
and 
\begin{equation}
\varphi= \frac{1}{2(2\pi)^2} \I(\R{f}).
\end{equation} 
Then 
\begin{equation}
f= \nabla \cdot E =- \Delta \varphi.
\label{ephif}
\end{equation}
\end{lemma}
\begin{proof}
\begin{eqnarray*}
 \nabla \cdot E &=&  - \frac{1}{2(2\pi)^2} \nabla \cdot \nabla \I(\R{f}) 
 = -\Delta \varphi  \nonumber \\
 &=&   \frac{1}{2(2\pi)^2}(-\Delta)  \I(\R{f}) 
 = f,
\end{eqnarray*}
where we used the inversion formula~(\ref{inversionformula_again}).
\end{proof}
Now we consider a measured tomogram $g : \PM^3 \to \RM$ and let us assume that $g=\R{f_0}$ for some $f_0 : \RM^3 \to \RM$. By Lemma \ref{lemma:normRf}-\ref{lemma:electricfield} it follows immediately that the fidelity term $\|\R{f}-g\|_{L^2(\PM^3)}^2$ can be rewritten as follows:
\begin{eqnarray}
\|\R{f}-g\|_{L^2(\PM^3)}^2 &= & 2(2\pi)^2 \|E-E_g\|^2_{L^2(\RM^3)} \nonumber \\
                                       & =& 2(2\pi)^2\|\nabla \varphi - \nabla \varphi_g\|^2_{L^2(\RM^3)},
\end{eqnarray}
where 
\begin{equation}
E=- \frac{1}{2(2\pi)^2} \nabla \I(\R{f}), \quad 
E_g=- \frac{1}{2(2\pi)^2} \nabla \I(g)
\end{equation} 
are the corresponding electric fields, while 
\begin{equation}
\varphi= \frac{1}{2(2\pi)^2} \I(\R{f}), \quad
\varphi_g= \frac{1}{2(2\pi)^2} \I(g)
\label{eq:potentials}
\end{equation} 
are the corresponding potentials.

With respect to the standard Mumford-Shah functional $\MS$ in (\ref{defn:JMS}), the new fidelity term in the functional $J_{\MS}$ in (\ref{def:MSfunctional}) controls the distance between the Radon transform of $f$ and the tomographic data $g$. The relevant difference with respect to the original functional is that the function $f$ and its Radon transform $\R{f}$ are defined in different spaces.
Let us try to interpret the  fidelity term $\|\R{f}-g\|_{L^2(\PM^3)}^2 $ from a physical point of view. A key ingredient for this goal is the electrostatics formulation of the Radon transform. This formulation can be summarized as follows: if we consider, in dimension $3$, a function $f$, 
we can think at it as a charge distribution density; if we apply to $f$ first the Radon operator $\R$ and then its adjoint $\I$ we obtain, up to a constant, the electrostatic potential generated by the charge distribution $f$. 
This formulation can be stated in any dimension $N$: the difference with general potential theory in dimension $N$ is that, in tomography, the potential produced by a point charge always scales like $1/|x|$, which is the case of electrostatic potential only in dimension $3$. From the electrostatic formulation of the Radon transform we can prove that the fidelity term in the functional $J_{\MS}$ actually imposes that the electric field produced by the charge distribution $f$ must be close to the ``measured electric field''. Therefore we conclude that the term $ \|\R{f}-g\|_{L^2(\PM^3)}^2$ is a fidelity term in this weaker sense.

Using this property based on the electrostatic interpretation of the tomographic reconstruction, we can try to minimize some appropriate functionals in the new variables $E$ (electric field) or $\varphi$ (electric potential) and then compute the corresponding  $f$ (charge density).
We manipulate the functional $J_{\MS}$ as follows:
\begin{eqnarray}\label{eqnaeeay_manipulation}
&& J_{\MS}(\Gamma,f)\nonumber \\
&=& \|\R{f}-g\|^2_{L^2(\PM^3)}+\alpha \int_{\RM^3 \setminus \Gamma} |\nabla f(x)|^2 \; \textrm{d}x + \beta \mathcal{H}^2(\Gamma) \nonumber \\
 &=&2 (2\pi)^2 \|E-E_g\|^2+\alpha \int_{\RM^3 \setminus \Gamma} |\nabla (\nabla \cdot E)(x)|^2 \; \textrm{d}x  \nonumber \\
 && + \beta \mathcal{H}^2(\Gamma) \nonumber \\
&=&2 (2\pi)^2 \|E-E_g\|^2+\alpha  \int_{\RM^3 \setminus \Gamma} | \Delta E|^2 \; \textrm{d}x  + \beta \mathcal{H}^2(\Gamma) \nonumber \\
&=& F(\Gamma,E) 
\end{eqnarray}
where $F$ is a new functional depending on a vector function $E$ and on a set $\Gamma$, and  we used the fact that $\nabla \wedge E =0$, since $E$ is conservative.

We observe that the functional
\begin{eqnarray}
F(\Gamma,E) &=& 2 (2\pi)^2  \|E-E_{g}\|^2+\alpha \int_{\RM^3 \setminus \Gamma} |\Delta E|^2 \; \textrm{d}x  
\nonumber\\
& & + \beta \mathcal{H}^2(\Gamma) 
\end{eqnarray}
is a second order functional for a vector field $E$ in which appears the measure of the set $\Gamma$ that is the set of discontinuities of $f$ and thus is the set of discontinuities of $\nabla \cdot E$. 
In the functional $F$ we recognize some similarities with a famous second-order free-discontinuity problem: the Blake-Zisserman model. This model is based on the minimization of the Blake-Zisserman functional
\begin{eqnarray}\label{defn:JBZ}
BZ(\Gamma_0,\Gamma_1,v)&=&\|v-v_0\|_{L^2(D)}^2  \nonumber \\
                                               && + \alpha \int_{D \setminus (\Gamma_0 \cup \Gamma_1)} |\Delta v(x)|^2\; \textrm{d}x \nonumber \\
                                              &&+ \beta \;\mathcal{H}^{N-1}(\Gamma_0 \cap D)\nonumber \\
                                              && + \gamma\;  \mathcal{H}^{N-1}((\Gamma_1 \setminus \Gamma_0) \cap D),
\end{eqnarray}
among admissible triplets $(\Gamma_0,\Gamma_1,v)$, where
\begin{itemize}
\item $D \subset \RM^N$ is an open set;
\item $\Gamma_0, \Gamma_1 \subset \RM^N$ are closed sets;
\item $\Gamma_0$ is the set of discontinuities of $v$ (jump set), and $\Gamma_1$ the the set of discontinuities of $\nabla v$ (crease set);
\item $v: D \to \RM$, $v \in C^2(D \setminus (\Gamma_0 \cup \Gamma_1)) \cap C(D \setminus \Gamma_0)$ is a scalar function;
\item $\Delta v$ denotes the distributional Laplacian of $v$;
\item $v_0 \in L^2(D)$ is the datum (grey intensity levels of the given image);
\item $\alpha, \beta, \gamma >0$ are parameters;
\item $\mathcal{H}^{N-1}$ denotes the $(N-1)$-dimensional Hausdorff measure. 
\end{itemize}
The Blake-Zisserman functional allows a more precise segmentation than the Mumford-Shah functional in the sense that also the curvature of the edges of the original picture is approximated. On the other hand, minimizers may not always exist, depending on the values of the parameters $\beta, \gamma$ and on the summability assumption on $v_0$.
We refer to \cite{BZ, CLT1, CLT2, CLT3, CLT4, CLT5} for motivation and analysis of variational approach to image segmentation and digital image processing. In particular see \cite{CLT1, CLT2, Coscia} for existence of minimizer results and \cite{CLT3} for a counterexample to existence and \cite{CLT6, CLT7}  for results concerning the regularity of minimizers.

Equation~(\ref{eqnaeeay_manipulation}) implies that the functional $J_{\MS}$ can be rewritten in terms of the vector field $E$ and of the discontinuities set of $\nabla \cdot E$, i.e. the set of creases of $E$, using the terminology of the Blake-Zisserman model. The fact that in the functional $F$ the discontinuities set of $E$ is not present depends on the fact that we are assuming that the charge density $f$ in the functional $J_{\MS}$ do not concentrate on surfaces or on lines. If we admit concentrated charge layers we can consider the Blake-Zisserman model  for the vector function $E$ as a relaxed version of the Mumford-Shah model for the charge $f$. In other words we propose to investigate the connections between minimizers of $J_{\MS}$ and minimizers of the higher order functional $J_{BZ}$: 
\begin{eqnarray}\label{defn:JBZE}
J_{BZ}(\Gamma_0,\Gamma_1,E)&=&\|E-E_g\|_{L^2(\RM^3)}^2  \nonumber \\
                                               && + \alpha \int_{\RM^3 \setminus (\Gamma_0 \cup \Gamma_1)} |\Delta E(x)|^2\; \textrm{d}x \nonumber \\
                                              &&+ \beta \;\mathcal{H}^{2}(\Gamma_0 )\nonumber \\
                                              && + \gamma\;  \mathcal{H}^{2}(\Gamma_1 \setminus \Gamma_0),
\end{eqnarray}
with the additional constraint $\nabla \wedge E=0$.
The main advantage of this approach is that the functional $J_{BZ}$  is a purely differential functional, while the functional $J_{\MS}$ is an integro-differential one. We expect that some results about the Blake-Zisserman model that could be rephrased into tomographic terms would provide immediately new results in tomography. Conversely all the peculiar tomographic features as the intrinsic vector nature of the variable $E$, the fact that its support cannot be bounded and the extra-constraint $\nabla \wedge E = 0$, motivate new research directions in the study of free-discontinuities problems. For example,  an interesting result in this context would be the determination of a good hypothesis on the datum $E_g$ that ensure that the charge density $f$ do not concentrate. \

We conclude this section with some comments:
\begin{itemize}
\item We proved that the measured data $g$ are actually the measured electric field produced by the unknown charge density, so  the term $\|\R{f}-g\|^2_{L^2(\PM^3)}$ in the functional  is a fidelity term in a weak sense.
\item The problem of the reconstruction of the charge can be rephrased into a reconstruction problem for the electric field. The electric field is an irrotational vector field, so the new minimization problem is actually a constrained minimization. In order to  avoid this constraint one could reformulate the reconstruction problem in terms of the electric potential $\varphi$ ($E=-\nabla \varphi$) obtaining a third-order functional  in which the fidelity term is 
\begin{equation}
\|\nabla \varphi - \nabla \varphi_g\|^2_{L^2(\RM^3)},
\end{equation}
where the potentials are given by~(\ref{eq:potentials}).
\item All this considerations hold true in dimension~$3$. In a generic dimension $n \geq 3$ the situation is quite different because the inversion formula for the Radon transform involves a (possibly fractional) power of the Laplacian. In this case the electrostatic description of tomography given in this section  fails. In order to restore it, it is necessary to consider another Radon-type transform which involves integrals of $f$ over linear manifolds with codimension $d$ such that $(n-d)/2=1$, i.e. $d= n-2$, see e.g. \cite{Helgason,tomolectures}.
\end{itemize}

\section*{Acknowledgments}
We thank G. Devillanova, G. Florio and F. Maddalena for for helpful discussions.

This work was  supported by ``Fondazione 
Cassa di Risparmio di Puglia'' and by the Italian National Group of Mathematical Physics (GNFM-INdAM).

\end{document}